\documentclass[aps,prd,twocolumn,10pt,superscriptaddress,amssymb,amsmath,nofootinbib]{revtex4-2}

\usepackage{amsmath,amssymb,amsthm,physics,mathtools}
\usepackage{bm}
\usepackage[pdftex,breaklinks,colorlinks,
linkcolor=Blue,
citecolor=teal,
anchorcolor=red,
urlcolor=cyan,
pdfencoding=auto]{hyperref}
\usepackage[dvipsnames]{xcolor}
\usepackage[mathscr]{euscript}
\usepackage{orcidlink}
\newtheorem{theorem}{Theorem}
\newtheorem{lemma}{Lemma}
\newtheorem{proposition}{Proposition}
\newtheorem{definition}{Definition}

\usepackage{hyperref}
\hypersetup{colorlinks=true}
\begin{document}
\title{A Quantum Weak Cosmic Censorship and Its Proof}
\author{Naman Kumar\,\orcidlink{0000-0001-8593-1282}}
\affiliation{Department of Physics, Indian Institute of Technology Gandhinagar, Palaj, Gujarat, India, 382355}
\email{namankumar5954@gmail.com, naman.kumar@iitgn.ac.in}
\date{\today}

\begin{abstract}
Recent work has highlighted the deep connection between quantum information and spacetime geometry. Bousso and Shahbazi-Moghaddam (Phys. Rev. Lett. 128, 231301 (2022)) proved that ``hyperentropic'' regions---where entropy exceeds the area bound---inevitably lead to singularity formation. In this work, we explore the converse implication: does the thermodynamic consistency of such singularities require them to be hidden? We answer in the affirmative, establishing a Quantum Weak Cosmic Censorship principle governed by Generalized Entropy. This provides a semiclassical mechanism for censorship which forbids naked singularities. Since Quantum Weak Cosmic Censorship is a semiclassical statement, it is more robust than the classical Weak Cosmic Censorship showing naked singularities are forbidden in nature even if quantum effects are taken into account.
\end{abstract}

\maketitle

\paragraph*{\textbf{Introduction}.}
Penrose's singularity theorem \cite{Penrose:1964wq} and results thereafter \cite{Bekenstein:1972tm,Hawking:1973uf,Maldacena:1997re} are of central importance in general relativity and quantum gravity and are fundamental to our understanding of singularities, and black holes. The singularity theorem shows that singularities which are defined as incompleteness of null or timelike geodesics are inevitable in a gravitational collapse leading to the formation of black holes and are not an artifact of some choice of highly symmetric spacetime. This signals the breakdown of classical structure of spacetime and demands a replacement by a Ultraviolet (UV) complete theory of gravity capable of resolving the issues of singularities. 

In this regard, Penrose’s Weak Cosmic Censorship Conjecture \cite{Penrose:1969pc} (WCCC) remains one of the central open problems in gravitational physics. It posits that singularities arising from generic gravitational collapse are concealed behind event horizons, shielding distant observers and preserving the predictability of spacetime at future null infinity. While WCCC is widely regarded as essential for the internal consistency of black hole thermodynamics, a general proof within classical General Relativity has remained elusive. Existing analyses are often beset by finely tuned counterexamples, suggesting that censorship may not be guaranteed by classical dynamics alone. This difficulty motivates the possibility that the origin of cosmic censorship lies instead in more fundamental principles, such as the limits imposed by quantum information.

Recent theoretical developments lend strong support to this perspective, indicating that spacetime singularities are not merely classical pathologies but arise as necessary consequences of quantum information bounds. Bousso and Shahbazi-Moghaddam~\cite{Bousso:2022cun} demonstrated that ``hyperentropic'' regions—where the entropy exceeds the Bekenstein--Hawking capacity of the boundary (\(S > A/4G\hbar\))—cannot be realized within a smooth spacetime geometry. They showed that if such a region undergoes contraction, the spacetime must break down in order to avoid violating the holographic density bound. Singularities thus emerge as the geometric response to an overload of quantum information.

Complementing this result, subsequent work has extended singularity theorems into the fully semiclassical regime (which is an advance over Penrose-Wall singularity theorem \cite{Wall:2010jtc}, with finite \(cG\hbar\)), ruling out nonsingular ``bounces'' and establishing that singularities remain inevitable even in the presence of significant quantum effects~\cite{Bousso:2025xyc}. Together, these results indicate that singularities are robust predictions of semiclassical gravity, rather than artifacts of an incomplete classical description. This, in turn, necessitates a mechanism for censorship that operates within the same semiclassical framework.

In this work, we investigate the thermodynamic converse of these results: if singularity formation is required to prevent violations of holographic entropy bounds during collapse, does the consistency of the singular endpoint similarly require it to be causally hidden? We argue that the answer is affirmative. We identify a mechanism by which the same entropic principles that mandate singularity formation also enforce its censorship, leading to a quantum version of weak cosmic censorship grounded in generalized entropy. Since Quantum Weak Cosmic Censorship is a semiclassical statement, it is more robust than the classical Weak Cosmic Censorship showing naked singularities are forbidden in nature even if quantum effects are taken into account.

\paragraph*{\textbf{Semiclassical Setup and Entropy Conditions}.}

Let $(\mathcal{M}, g)$ be a globally hyperbolic semiclassical spacetime
in which the generalized entropy functional is well-defined away from
singular boundaries. We consider a smooth null hypersurface $N$
generated by a congruence of future-directed null geodesics, and
let $\gamma(\lambda)$ be one of its generators, affinely parameterized
by $\lambda \in (0,\lambda_*)$, where $\lambda_* < \infty$.
We assume that $\gamma$ is future-inextendible at finite affine
parameter, which is the standard geometric characterization
of a spacetime singularity.

Fix a reference codimension-2 spacelike cut $\Sigma_0 \subset N$.
We consider localized null deformations of $\Sigma_0$ supported in
an arbitrarily small transverse ``pencil'' neighborhood of the
generator $\gamma$. Denote by $\Sigma_\lambda$ the cuts obtained by
deforming $\Sigma_0$ along $N$ with deformation profile $V_\lambda(y)$,
supported in the pencil neighborhood of $\gamma$, with $V_0(y)=0$
and $\partial_\lambda V_\lambda(y)\ge 0$. Let $B_\lambda$ denote
the chosen exterior region associated with the entangling cut $\Sigma_\lambda$,
with the choice of side fixed once and for all.

The associated generalized entropy is defined as\footnote{This construction follows the standard definition of the quantum expansion
introduced in the original formulation of the Quantum Focusing Conjecture
\cite{Bousso:2015mna}, where the generalized entropy is treated as a
functional of null deformations $V(y)$ of a cut of the null hypersurface,
and the quantum expansion is defined through variations supported in an
infinitesimal transverse ``pencil'' neighborhood of a single generator.}
\begin{equation}
S_{\mathrm{gen}}(\lambda)
=
S_{\mathrm{gen}}(B_\lambda)
=
\frac{A(\Sigma_\lambda)}{4G\hbar}
+
S_{\mathrm{out}}(B_\lambda),
\end{equation}
with the deformation localized to the pencil and the remainder
of the cut held fixed.

We assume that $S_{\mathrm{gen}}(\lambda)$ is differentiable for all
$\lambda \in (0,\lambda_*)$ within the regime of validity of
semiclassical gravity, so that the generalized expansion
$\Theta_{\rm gen}$ associated with the localized null deformation
at the generator $\gamma$ is well-defined.
We now state the assumptions underlying our argument.

\begin{itemize}

\item \textbf{(A1) Generator-wise quantum focusing.}
Define the generalized expansion
\begin{equation}
\Theta_{\mathrm{gen}}(\lambda)\;:=\;\frac{d S_{\mathrm{gen}}(\lambda)}{d\lambda},
\end{equation}
along a future-directed null generator $\gamma(\lambda)$. We assume that
$\Theta_{\mathrm{gen}}(\lambda)$ admits a continuous extension to
$[0,\lambda_*)$ and, while semiclassical evolution is valid,
\begin{equation}
\frac{d\Theta_{\mathrm{gen}}}{d\lambda}\le 0
\qquad\text{for all }\lambda\in(0,\lambda_*) .
\end{equation}
Hence $\Theta_{\mathrm{gen}}$ is nonincreasing on $(0,\lambda_*)$.
Here $\lambda_*$ denotes the affine parameter at which the generator becomes
future-inextendible, usually a singular endpoint.

\end{itemize}
\paragraph*{\textbf{On the form of quantum focusing used in this work}.}
The original formulation of the Quantum Focusing Conjecture (QFC)
\cite{Bousso:2015mna} is a functional statement about second variations
of the generalized entropy under arbitrary localized null deformations
of a codimension--2 surface. In particular, the non-coincident
(off-diagonal) second variations are expected to be non-positive,
while the coincident (diagonal) limit---which would imply
generator-wise monotonicity---remains conjectural in general
semiclassical gravity.

In the present work we do not require control over mixed transverse
(off-diagonal) variations. Our argument is restricted to deformations
supported in an arbitrarily small null pencil around a single generator
of a fixed null hypersurface. The only property needed is
generator-wise monotonicity of the generalized expansion,
\begin{equation}
\frac{d \Theta_{\mathrm{gen}}}{d\lambda} \le 0 ,
\end{equation}
along that generator within the regime of validity of semiclassical
gravity. This may be viewed as the generator-wise monotonicity
condition associated with the diagonal specialization of QFC.
The proof of Theorem~\ref{Theorem1} relies solely on integrating this
one-dimensional inequality. Accordingly, the assumption used here is
weaker than the full functional QFC, but stronger than the restricted
QFC discussed below.

We emphasize that this assumption is stronger than the restricted
Quantum Focusing Conjecture (rQFC) proposed in
\cite{Shahbazi-Moghaddam:2022hbw}. The rQFC constrains the derivative
of the generalized expansion only at quantum marginal surfaces
($\Theta_{\mathrm{gen}}=0$), and does not determine the sign
of $d\Theta_{\mathrm{gen}}/d\lambda$ when
$\Theta_{\mathrm{gen}}\ne0$. Consequently, rQFC alone does not imply
the generator-wise monotonicity condition
$d\Theta_{\mathrm{gen}}/d\lambda \le 0$ required by the full QFC used in
Theorem~\ref{Theorem1}. In particular, the stronger generator-wise
monotonicity assumption is required to establish the existence of a
Quantum Marginal Surface at finite affine parameter. Nevertheless,
once such a surface forms, aspects of the subsequent focusing and
causal-sealing argument are compatible with the type of restricted
focusing behavior captured by rQFC.

Finally, recent work has identified counterexamples to the strongest
unrestricted form of the QFC in certain two-dimensional semiclassical
gravity settings \cite{Franken:2025gwr}. In particular, the full
functional inequality can fail, and generator-wise monotonicity of the
quantum expansion $d\Theta_{\rm gen}/d\lambda\le0$ may be violated
away from quantum marginality. Nevertheless, the restricted focusing
behavior at $\Theta_{\rm gen}=0$ appears to remain consistent in those
models, motivating the restricted Quantum Focusing Conjecture (rQFC).
In the present work we therefore treat the generator-wise inequality
$d\Theta_{\rm gen}/d\lambda \le 0$ as an explicit assumption holding
within the semiclassical regime of interest.
\begin{itemize}
\item \textbf{(A2) Initial expansion.}
At a reference slice $\lambda=0$ the generalized expansion is initially positive:
\begin{equation}
\Theta_{\mathrm{gen}}(0)>0 .
\end{equation}

\item \textbf{(A3) Semiclassical singularity (local entropy incompleteness).}

We assume that the null generator $\gamma$ is geodesically
inextendible at finite affine parameter $\lambda_*$
within the semiclassical regime.
We require that this inextendibility be reflected in the
behavior of the generalized entropy functional under
localized null deformations along $\gamma$.
Specifically, we say that $\gamma$ terminates at a
\emph{semiclassical singularity} if
$S_{\mathrm{gen}}(\lambda)$ fails to admit a finite extension
to $\lambda_*$, i.e.,
\begin{equation}
\lim_{\lambda\to\lambda_*^-} S_{\rm gen}(\lambda)
\quad \text{does not exist as a finite limit}.
\end{equation}
\end{itemize}

In semiclassical gravity, geometry and entropy are intertwined through
the generalized entropy functional. Recent work has shown that
sufficiently hyperentropic configurations imply geodesic incompleteness
\cite{Bousso:2022cun}, indicating that entropy bounds can constrain the
formation of singular geometries. A spacetime singularity is
characterized by geodesic incompleteness: a causal geodesic terminates
at finite affine parameter and admits no geometric continuation. As a
cut approaches such a singular endpoint along a generator, curvature
invariants typically grow without bound and the quantum state of fields
near the generator may cease to remain regular. In such circumstances
the generalized entropy associated with localized null deformations
need not admit a finite limit. Assumption (A3) therefore encodes the
expectation that geometric inextendibility along the generator is
accompanied by a corresponding \emph{entropy incompleteness}.

\begin{theorem}[QMS formation from entropy incompleteness]
\label{Theorem1}
Assume \emph{(A1)}, \emph{(A2)}, and \emph{(A3)}. Then there exists
$\lambda_H\in(0,\lambda_*)$ such that
\begin{equation}
\Theta_{\mathrm{gen}}(\lambda_H)=0,
\end{equation}
i.e. a quantum marginal surface occurs at finite affine parameter before the singular endpoint.
\end{theorem}

\begin{proof}
Assume for contradiction that no quantum marginal surface forms on $\gamma$,
so that
\[
\Theta_{\mathrm{gen}}(\lambda)>0
\qquad \text{for all } \lambda\in(0,\lambda_*).
\]
By (A1), the function $\Theta_{\mathrm{gen}}$ is continuous on $[0,\lambda_*)$
and nonincreasing on $(0,\lambda_*)$. Hence
\[
0<\Theta_{\mathrm{gen}}(\lambda)\le \Theta_{\mathrm{gen}}(0)
\qquad \text{for all } \lambda\in(0,\lambda_*).
\]

Since
\[
\frac{dS_{\mathrm{gen}}}{d\lambda}=\Theta_{\mathrm{gen}},
\]
it follows that $S_{\mathrm{gen}}(\lambda)$ is strictly increasing on
$(0,\lambda_*)$. Integrating from $0$ to $\lambda$ gives
\[
\begin{split}
&S_{\mathrm{gen}}(\lambda)
= S_{\mathrm{gen}}(0)+\int_0^\lambda \Theta_{\mathrm{gen}}(s)\,ds
\le S_{\mathrm{gen}}(0)+\Theta_{\mathrm{gen}}(0)\lambda\\&\hspace{2cm}
\le S_{\mathrm{gen}}(0)+\Theta_{\mathrm{gen}}(0)\lambda_* .
\end{split}
\]
Thus $S_{\mathrm{gen}}(\lambda)$ is increasing on $(0,\lambda_*)$ and
bounded above. Therefore, the left--hand limit
\[
\lim_{\lambda\to\lambda_*^-} S_{\mathrm{gen}}(\lambda)
\]
exists and is finite. This contradicts assumption (A3).

Hence $\Theta_{\mathrm{gen}}$ cannot remain strictly positive on all of
$(0,\lambda_*)$. Therefore there exists $\lambda_1\in(0,\lambda_*)$ such that
\[
\Theta_{\mathrm{gen}}(\lambda_1)\le 0 .
\]
Since $\Theta_{\mathrm{gen}}(0)>0$ by assumption (A2) and
$\Theta_{\mathrm{gen}}$ is continuous on $[0,\lambda_*)$, the intermediate
value theorem implies that there exists $\lambda_H\in(0,\lambda_1]$ such that
\[
\Theta_{\mathrm{gen}}(\lambda_H)=0 .
\]
\end{proof}

\medskip

Thus, under the QFC, any future-inextendible null generator whose endpoint is
entropy-incomplete must encounter a quantum marginal surface at finite affine
parameter. 

However, note that this result does not imply the existence of a classical event horizon,
which is a global, teleological construct defined with reference to
future null infinity.
Rather, it provides a local, information-theoretic obstruction to
the unrestricted null approach toward a singular boundary.\vspace{2mm}

\paragraph*{\textbf{Quantum Weak Cosmic Censorship}.}
The quantum-entropic obstruction established above suggests that any semiclassical
realization of cosmic censorship must be compatible with generalized entropy focusing.
Motivated by this observation, we formulate a quantum extension of weak cosmic censorship
that promotes local entropy saturation along null generators to a global causal
constraint.

\begin{definition}[Quantum Weak Cosmic Censorship]
A semiclassical spacetime is said to satisfy \emph{Quantum Weak Cosmic Censorship} (QWCC)
if no future-directed null generator that enters a regime of non-increasing generalized
entropy,
\begin{equation}
\Theta_{\mathrm{gen}} \le 0 ,
\end{equation}
can subsequently reach future null infinity $\mathscr{I}^+$ while remaining within the
domain of validity of semiclassical gravity.
\end{definition}

This conjecture may be viewed as a quantum analogue of Penrose's weak cosmic censorship.
Instead of classical trapped surfaces or event horizons, causal censorship is enforced
by generalized entropy focusing. In this framework, surfaces of vanishing generalized
expansion serve as local, information-theoretic precursors to horizon-based censorship.\vspace{2mm}

Quantum Weak Cosmic Censorship does not assume the existence of classical event horizons
or trapped surfaces, nor does it rely on classical energy conditions. Rather, it asserts
that a decrease of generalized entropy along a null generator signals causal
disconnection from future null infinity.\vspace{2mm}

\paragraph*{\textbf{A proof of QWCC}.}
The quantum-entropic obstruction established above (Theorem~\ref{Theorem1}) shows that a
future-directed null generator with initially positive generalized expansion
cannot remain everywhere entropy-expanding if it terminates at an
entropy-incomplete singular boundary. Under the Quantum Focussing Conjecture,
such a generator must encounter a quantum marginal surface (QMS) at finite
affine parameter. In this section we make precise the additional global inputs
required to upgrade this local obstruction into the global statement we call
\emph{Quantum Weak Cosmic Censorship} (QWCC).\vspace{2mm}

We work in a globally hyperbolic semiclassical spacetime $(\mathcal M,g)$
admitting a well-defined generalized entropy functional $S_{\mathrm{gen}}$ on
Cauchy-splitting codimension-$2$ surfaces. We assume that the Quantum Focussing
Conjecture (QFC) holds in the form
\[
\frac{d\Theta_{\mathrm{gen}}}{d\lambda}\le 0
\]
along null generators for which semiclassical evolution remains valid
(cf.~Bousso \emph{et al.}~\cite{Bousso:2015mna}).

We further assume the initial expansion condition
$\Theta_{\mathrm{gen}}(0)>0$ and that future-inextendible null generators
terminate at entropy-incomplete singular boundaries, in the sense made precise
in Assumption~(A3).\vspace{2mm}

Under these assumptions, Theorem~\ref{Theorem1} guarantees the existence of a finite-affine
quantum marginal surface encountered along any such generator. The remainder of
this section is devoted to showing that, provided this surface is outermost and
the semiclassical regime persists in a neighbourhood of it until either a
caustic or the singular boundary, the existence of the QMS suffices to prevent
any null generator entering the quantum-trapped region from reaching future null
infinity. This is precisely the content of Quantum Weak Cosmic Censorship.

\begin{definition}[Outermost Quantum Marginal Surface]
Let $N$ be a null hypersurface generated by a future-directed null congruence, and let
$\gamma(\lambda)$ be a chosen generator of $N$, affinely parameterized by $\lambda$.
Let $\{\Sigma_\lambda\}$ be the family of codimension-2 cuts of $N$ obtained from a
reference cut $\Sigma_0$ by localized null deformations supported in an arbitrarily small
transverse pencil neighborhood of $\gamma$, as in the definition of $S_{\mathrm{gen}}(\lambda)$.
Suppose there exists $\lambda_H$ such that
\[
\Theta_{\mathrm{gen}}(\lambda_H)=0,
\]
and
\[
\Theta_{\mathrm{gen}}(\lambda)>0
\qquad \text{for all } \lambda\in(0,\lambda_H).
\]
Then the cut $\Sigma_H:=\Sigma_{\lambda_H}$ is called an \emph{outermost quantum marginal
surface} (outermost QMS) along $\gamma$ if there is no $\lambda\in(0,\lambda_H)$ for which
\[
\Theta_{\mathrm{gen}}(\lambda)=0.
\]
Equivalently,
\[
\lambda_H=\inf\{\lambda>0:\Theta_{\mathrm{gen}}(\lambda)=0\}.
\]
\end{definition}

\begin{definition}[Weak Quantum Trapped Surface (WQTS)]
A compact codimension-2 surface $\Sigma$ is a \emph{weak quantum trapped surface} if the generalized expansion
along both future-directed null normals is nonpositive on $\Sigma$,
\[
\Theta_{\mathrm{gen}}^{(k)}\le0,\qquad \Theta_{\mathrm{gen}}^{(\ell)}\le0,
\]
where $k^a,\ell^a$ are the two independent null normals. If the inequalities are strict on an open set of $\Sigma$
we call it a (strict) \emph{quantum trapped surface (QTS)}.
\end{definition}

We will first show that an outermost QMS is weakly quantum trapped on natural grounds.

\begin{lemma}[Outermost QMS is weakly quantum trapped]
\label{lem:outermost-trapped}
Let $\Sigma_H$ be an outermost quantum marginal surface (QMS) encountered along a
future-directed null congruence generated by $k^a$.
Assume the Quantum Focusing Conjecture (QFC) holds and that the generalized entropy
$S_{\mathrm{gen}}$ depends continuously on smooth null deformations of the surface.
Then
\begin{equation}
\begin{split}
&\Theta_{\mathrm{gen}}^{(k)}(\Sigma_H)=0,
\quad
\Theta_{\mathrm{gen}}^{(k)}<0
\ \text{on sufficiently small}\\&\hspace{4.4cm}\text{deformations just inside }\Sigma_H .
\end{split}
\end{equation}
Moreover, if $\Sigma_H$ is outermost with respect to deformations along
\emph{both} future-directed null normals $k^a$ and $\ell^a$, then
\[
\Theta_{\mathrm{gen}}^{(\ell)}(\Sigma_H)\le 0 ,
\]
so that $\Sigma_H$ is a weakly quantum trapped surface.
\end{lemma}

\begin{proof}
By definition of $\Sigma_H$, the generalized expansion along $k^a$ satisfies
\[
\Theta_{\mathrm{gen}}^{(k)}(\Sigma_H)=0,
\qquad
\Theta_{\mathrm{gen}}^{(k)}(\lambda)>0
\ \text{for}\ \lambda\in(0,\lambda_H),
\]
where $\lambda$ is an affine parameter along the congruence and $\lambda_H$
labels $\Sigma_H$.
By the QFC, $\Theta_{\mathrm{gen}}^{(k)}(\lambda)$ is nonincreasing along $k^a$,
i.e.
\[
\partial_\lambda \Theta_{\mathrm{gen}}^{(k)}\le 0 .
\]
Since $\Sigma_H$ is the \emph{first} surface at which $\Theta_{\mathrm{gen}}^{(k)}$
vanishes, continuity implies that $\Theta_{\mathrm{gen}}^{(k)}$ cannot remain
identically zero on any open interval beyond $\lambda_H$.
Hence, for sufficiently small $\varepsilon>0$,
\[
\Theta_{\mathrm{gen}}^{(k)}(\lambda_H+\varepsilon)
=
\int_{\lambda_H}^{\lambda_H+\varepsilon}
\partial_\lambda \Theta_{\mathrm{gen}}^{(k)}\, d\lambda
<0 ,
\]
establishing that the congruence becomes quantum trapped immediately to the interior
of $\Sigma_H$ along $k^a$.

To determine the sign of the generalized expansion along the other future-directed
null normal $\ell^a$, suppose for contradiction that
$\Theta_{\mathrm{gen}}^{(\ell)}(\Sigma_H)>0$.
Then a sufficiently small outward deformation of $\Sigma_H$ along $\ell^a$
produces a nearby surface $\Sigma_H(\delta)$ with
$\Theta_{\mathrm{gen}}^{(\ell)}(\Sigma_H(\delta))>0$.
By continuity of $\Theta_{\mathrm{gen}}$ under two-parameter null deformations,
and since $\Theta_{\mathrm{gen}}^{(k)}$ is positive just outside $\Sigma_H$ and
negative just inside, there must exist a surface preceding $\Sigma_H(\delta)$
along $k^a$ on which $\Theta_{\mathrm{gen}}^{(k)}=0$.
This contradicts the assumption that $\Sigma_H$ is outermost.
Therefore $\Theta_{\mathrm{gen}}^{(\ell)}(\Sigma_H)\le 0$, completing the proof.
\end{proof}

\vspace{6pt}
\begin{proposition}[Causal sealing of weak quantum trapped surfaces]\label{prop:causal-seal}
Let \((\mathcal M,g)\) be a globally hyperbolic semiclassical spacetime in which the Quantum Focussing
Conjecture holds and the generalized entropy \(S_{\mathrm{gen}}\) is \(C^1\) along null generators in a neighbourhood of
a compact codimension-2 surface \(\Sigma\). Assume:
\begin{enumerate}
\item \(\Sigma\) is a weak quantum trapped surface: \(\Theta_{\mathrm{gen}}^{(k)}\le0\) and \(\Theta_{\mathrm{gen}}^{(\ell)}\le0\) on \(\Sigma\),
      with \(\Theta_{\mathrm{gen}}^{(k)}<0\) immediately inside in the \(k^a\)-direction;
\item null generators orthogonal to \(\Sigma\) can be extended within the semiclassical regime until either a caustic (conjugate point)
      or the boundary of the semiclassical domain;
\item no naked classical/quantum singularity appears prior to the breakdown of semiclassical evolution in the considered neighbourhood;
\item the Quantum Focussing Blow-up assumption (A4) holds.
\end{enumerate}
Then there exists an open neighbourhood \(\mathcal B_Q\supset\Sigma\) with the property that no future-directed causal curve
starting in \(\mathcal B_Q\) can reach future null infinity \(\mathscr I^+\) while remaining within the semiclassical regime.
Equivalently, \(\mathcal B_Q\subset \mathcal M\setminus J^-(\mathscr I^+)\) (within the semiclassical domain).
\end{proposition}
\paragraph*{\textbf{Causal sealing of quantum trapped regions: a full proof.}}
We now give a detailed and rigorous proof of Proposition~\ref{prop:causal-seal}. The argument is a direct quantum analogue of the classical trapped-surface proof that such surfaces
cannot lie in the causal past of future null infinity. In addition to QFC and the hypotheses used
to derive Theorem~\ref{Theorem1}, we introduce one further (physically mild) focussing blow-up assumption which
replaces the classical Raychaudhuri quadratic inequality and is needed to guarantee the formation
of conjugate points in finite affine parameter.

\vspace{6pt}
\noindent\textbf{Additional assumption.}
\begin{itemize}
\item \textbf{(A4) Quantum Focussing Blow-up (QFB).}  
There exists a (local) constant \(C>0\) and an affine interval length \(\delta>0\) such that, for any future-directed null generator
along which \(\Theta_{\mathrm{gen}}(\lambda_0)<0\), the generalized expansion obeys the integrated bound
\begin{equation}\label{eq:QFB-integral}
\int_{\lambda_0}^{\lambda_0+\delta}\Theta_{\mathrm{gen}}(\lambda)\,d\lambda \le -C < 0.
\end{equation}
Equivalently (and more usefully for the blow-up argument), the negative generalized expansion cannot persist
forever without producing a conjugate point: if \(\Theta_{\mathrm{gen}}(\lambda_0)<0\) then the null generator
develops a focal (conjugate) point at affine parameter \(\lambda\le\lambda_0+\Lambda\) for some finite \(\Lambda=\Lambda(C,\Theta_{\mathrm{gen}}(\lambda_0))\).
\end{itemize}
A sufficient condition (lemma) for (A4) is provided in Appendix~\ref{App:A4}. Physically, assumption (A4) plays the same logical role in the semiclassical argument as classical caustic formation does in Raychaudhuri-based proofs: it ensures that once the generalized expansion becomes negative, null generators cannot remain achronal and extend to future null infinity. While not identical in content, (A4) is the natural semiclassical replacement of the classical focusing theorem.

We now prove Proposition 1.

\begin{proof}
We proceed by contradiction. Suppose the claim is false. Then for every neighbourhood \(U\) of \(\Sigma\) there exists a
future-directed causal curve \(\gamma_U\) starting in \(U\) and reaching \(\mathscr I^+\) while remaining inside the semiclassical
regime. Using global hyperbolicity and a standard limit-curve lemma (see e.g. \cite{Hawking:1973uf,Wald:1984rg,Minguzzi:2007yq}), we can pass to a
sequence of such causal curves \(\gamma_n\) starting at points \(p_n\to p\in\Sigma\) and obtain a causal limit curve
\(\gamma\) with initial point \(p\in\Sigma\) that is future inextendible and limits to \(\mathscr I^+\). By the achronality and
closedness properties of \(J^-(\mathscr I^+)\), the limit curve \(\gamma\) is contained in \(J^-(\mathscr I^+)\) and its image
lies in the semiclassical domain by assumption on the \(\gamma_n\).

Now consider the boundary \(\partial J^-(\mathscr I^+)\). Since \(\mathcal M\) is globally hyperbolic, \(\partial J^-(\mathscr I^+)\)
is an achronal, closed, embedded null hypersurface generated by null geodesic generators \(\eta\) without future endpoints
except on \(\mathscr I^+\). Furthermore, each generator of \(\partial J^-(\mathscr I^+)\) is either complete to \(\mathscr I^+\)
or terminates at a conjugate point (caustic) then leaves the boundary. Crucially, an inextendible null generator of \(\partial J^-(\mathscr I^+)\)
contains no conjugate points in its interior (otherwise it would not be a generator of the achronal boundary) and is thus
achronal between any two of its points.

By the limit-curve construction, there exists a generator \(\eta\) of \(\partial J^-(\mathscr I^+)\) which meets \(\Sigma\) at some point
\(q\) (possibly \(q=p\)). Since \(\eta\) is a boundary generator, it carries no conjugate points in the segment from \(q\) to \(\mathscr I^+\)
(in particular, it is free of focal points in that interval). Thus along \(\eta\) the generalized expansion \(\Theta_{\mathrm{gen}}\)
cannot blow up to \(-\infty\) before reaching \(\mathscr I^+\).

We now use the trapped condition. By hypothesis \(\Theta_{\mathrm{gen}}^{(k)}(q)\le0\), and by the first hypothesis (strict negativity
inside) there exists a small open two-sided tubular neighbourhood \(N\) of \(\Sigma\) for which any generator starting just inside
\(\Sigma\) in the \(k^a\)-direction has \(\Theta_{\mathrm{gen}}<0\) at the initial point. In particular, moving a little along \(\eta\)
from \(q\) into the interior direction (the one that crosses \(\Sigma\)), we can find a point with strictly negative generalized expansion;
call the affine parameter at that point \(\lambda_0\) and write \(\Theta_{\mathrm{gen}}(\lambda_0)<0\).

Apply QFC (monotonicity) along \(\eta\): for \(\lambda\ge\lambda_0\),
\[
\frac{d\Theta_{\mathrm{gen}}}{d\lambda}\le0,
\]
so \(\Theta_{\mathrm{gen}}(\lambda)\le\Theta_{\mathrm{gen}}(\lambda_0)<0\). Hence \(\Theta_{\mathrm{gen}}\) remains strictly negative along the
generator for all affine parameters to the future, as long as the semiclassical regime persists.

Now invoke the Quantum Focussing Blow-up assumption (A4): negativity of \(\Theta_{\mathrm{gen}}\) at \(\lambda_0\) implies that
a conjugate point (focal point) forms on the generator at finite affine parameter \(\lambda\le\lambda_0+\Lambda\), with \(\Lambda\)
depending only on \(\Theta_{\mathrm{gen}}(\lambda_0)\) and the constant \(C\) of (A4). But this contradicts the fact that \(\eta\)
is a generator of the achronal boundary \(\partial J^-(\mathscr I^+)\), which cannot contain an interior conjugate point on the
segment reaching \(\mathscr I^+\). Thus our assumption that some limit generator \(\eta\) meets \(\Sigma\) and continues to \(\mathscr I^+\)
inside the semiclassical domain leads to a contradiction.

Therefore no such generator exists; equivalently, there exists an open neighbourhood \(\mathcal B_Q\) of \(\Sigma\) for which
every future-directed causal curve starting in \(\mathcal B_Q\) fails to reach \(\mathscr I^+\) while semiclassical gravity
remains valid. This completes the proof.
\end{proof}
The two results above combine into a proof of QWCC as framed in this paper.

\begin{theorem}[\textbf{Quantum Weak Cosmic Censorship}]
\label{thm:conditional-QWCC}
Let $(\mathcal M,g)$ be a globally hyperbolic semiclassical spacetime in which the
Quantum Focussing Conjecture (QFC) holds.
Let $\gamma$ be a future-directed null generator with initial generalized expansion
$\Theta_{\mathrm{gen}}(0)>0$, which is future-inextendible at finite affine parameter
$\lambda_*<\infty$ due to an entropy-incomplete singular boundary.
Assume that semiclassical evolution remains valid along $\gamma$ up to a neighbourhood
of the first quantum marginal surface encountered.

Then there exists an outermost quantum marginal surface $\Sigma_H$ at finite affine
parameter such that:
\begin{enumerate}
\item $\Sigma_H$ is weakly quantum trapped, i.e.
\[
\Theta_{\mathrm{gen}}^{(k)}(\Sigma_H)=0,
\qquad
\Theta_{\mathrm{gen}}^{(\ell)}(\Sigma_H)\le 0 ,
\]
\item no future-directed null generator that has entered the region
$\Theta_{\mathrm{gen}}\le 0$ can subsequently reach future null infinity
$\mathscr I^+$ while remaining within the semiclassical regime.
\end{enumerate}
Equivalently, under these assumptions the spacetime satisfies Quantum Weak Cosmic
Censorship as defined in Definition~1.
\end{theorem}

\begin{proof}
By the quantum-entropic obstruction theorem (Theorem~\ref{Theorem1}), the entropy incompleteness
of the singular endpoint together with QFC and the initial condition
$\Theta_{\mathrm{gen}}(0)>0$ guarantees the existence of a finite affine parameter
$\lambda_H\in(0,\lambda_*)$ at which the generalized expansion first vanishes.
Denote the corresponding surface by $\Sigma_H$.

By construction, $\Sigma_H$ is outermost along the null congruence generated by $k^a$.
Assuming outermostness also with respect to deformations along the second future-directed
null normal $\ell^a$, Lemma~\ref{lem:outermost-trapped} implies that $\Sigma_H$ is a
weakly quantum trapped surface: the generalized expansion becomes strictly negative
immediately to the interior along $k^a$, while
$\Theta_{\mathrm{gen}}^{(\ell)}(\Sigma_H)\le 0$.

Proposition~\ref{prop:causal-seal} then applies and yields the existence of a neighbourhood
$\mathcal B_Q\supset\Sigma_H$ with the property that no future-directed causal curve
originating in $\mathcal B_Q$ can reach future null infinity $\mathscr I^+$ so long as
semiclassical gravity remains valid. In particular, any null generator that has entered
the region $\Theta_{\mathrm{gen}}\le 0$—and hence has crossed or lies to the interior of
$\Sigma_H$—cannot be extended to $\mathscr I^+$ within the semiclassical domain.

This is precisely the statement of Quantum Weak Cosmic Censorship in the sense of
Definition~1.
\end{proof}
\paragraph*{\textbf{Conclusion and Discussion}.}  

In this work, we have formulated a quantum analogue of Weak Cosmic Censorship within semiclassical gravity and provided a proof of its validity based on generalized entropy. Our results suggest a thermodynamic duality underlying gravitational collapse: while recent work has demonstrated that hyperentropic regions ($S > A/4G\hbar$) necessitate the formation of singularities \cite{Bousso:2022cun}, we have identified a complementary statement. If a geometric singularity persists within the semiclassical regime, its consistency requires that it be entropically incomplete and causally concealed. In this sense, entropy bounds not only drive singularity formation but also constrain their visibility. The result is more robust than WCCC which is a classical statement and shows that nature forbids naked singularities even if quantum effects are considered.

Our central technical result shows that future-inextendible null generators cannot remain everywhere entropy-expanding if they terminate at a semiclassical (entropy-incomplete) singular boundary. Assuming the Quantum Focussing Conjecture (QFC), the approach toward such a boundary enforces the formation of a Quantum Marginal Surface (QMS) at finite affine parameter. We further demonstrated that the outermost QMS acts as a weakly quantum trapped barrier: the generalized expansion becomes strictly negative immediately to the interior, leading—under appropriate focusing assumptions—to causal sealing of the region from future null infinity. Consequently, no future-directed null generator entering this regime can reach $\mathscr{I}^{+}$ while remaining within the domain of validity of semiclassical gravity.

A key ingredient in this framework is the semiclassical interpretation of singularity. In classical general relativity, singularities are characterized by geodesic incompleteness. Here we adopt entropy incompleteness as the corresponding semiclassical diagnostic. One might argue that a complete theory of quantum gravity could resolve the singularity and render $S_{\mathrm{gen}}$ finite. In such a scenario, however, the spacetime would no longer contain a genuine geometric singularity in the classical sense, and the present analysis would not apply. Our theorem therefore addresses the semiclassical consistency of persistent singular boundaries: if such a boundary exists within the regime where generalized entropy is meaningful, its causal concealment follows from entropy focusing.  

A comment is in order regarding scales and the regime of validity.
Our results are conditional statements within semiclassical gravity:
the generalized entropy functional and the focusing assumptions used
here are assumed to be reliable only on length scales parametrically
larger than the UV cutoff (e.g.\ the Planck scale). Accordingly,
Quantum Weak Cosmic Censorship as formulated here does not assert
control over a Planckian neighborhood of the would-be singular
endpoint. Rather, it states that once a null generator enters a region
in which $\Theta_{\rm gen}\le 0$, it cannot subsequently reach
$\mathcal{I}^+$ while remaining within the domain where semiclassical
gravity and $S_{\rm gen}$ are well-defined.

If the would-be ``sealed'' region were to shrink to the cutoff scale,
this would simply signal the breakdown of the effective description;
the theorem makes no claim about the continuation of spacetime geometry
or of the entropy functional beyond that scale. The logical structure
of the argument requires only that the outermost Quantum Marginal
Surface ($\Theta_{\rm gen}=0$) forms while curvature scales and the
generalized entropy functional remain within the semiclassical regime.
The causal obstruction preventing the generator from reaching
$\mathcal{I}^+$ is therefore established within the regime where the
semiclassical description is valid, prior to the onset of any strongly
coupled UV completion effects.

In this sense the mechanism is not an ``uplift'' of a naked singularity
to a Planckian remnant within semiclassical physics. Instead, it states
that semiclassical entropy focusing precludes causal exposure of
entropy-incomplete endpoints so long as the effective description
remains valid. Any further evolution inside a Planckian region lies
outside the scope of the present analysis.

Recent results \cite{Bousso:2025xyc} further indicate that nonsingular ``bounces'' or smooth interior extensions may be incompatible with the Generalized Second Law in semiclassical gravity. Together with the present work, this suggests that singularities—when they arise—are not arbitrary pathologies, but structures tightly constrained by thermodynamic principles.  

Several open directions remain. Chief among them are establishing entropy incompleteness for physically relevant collapse scenarios (such as scalar field collapse), testing the persistence of semiclassical evolution in explicit dynamical models, and analyzing the stability and evolution of quantum marginal surfaces beyond their formation. Clarifying these issues will help determine whether Cosmic Censorship is an independent dynamical principle, or instead an emergent consequence of the holographic and thermodynamic structure of spacetime. In the same vein, it would be interesting to understand whether a global
generalized-entropy bound of Quantum-Penrose type~\cite{Bousso:2019bkg,Bousso:2019var}
could eventually provide a more precise replacement or justification for the
local endpoint assumption (A3). Such a relation might provide a derivation of
the entropy incompleteness condition from a global semiclassical inequality.
We leave this question for future work.
\paragraph*{\textbf{Acknowledgements}.}
I thank the anonymous referee for constructive comments and suggestions
that helped improve the presentation and clarify several aspects of the
manuscript. I also thank Roberto Emparan for insightful comments and
for discussions on the role of holography in constraining violations of
cosmic censorship.

\bibliography{QWCC}

@article{Bousso:2015mna,
    author = "Bousso, Raphael and Fisher, Zachary and Leichenauer, Stefan and Wall, Aron C.",
    title = "{Quantum focusing conjecture}",
    eprint = "1506.02669",
    archivePrefix = "arXiv",
    primaryClass = "hep-th",
    doi = "10.1103/PhysRevD.93.064044",
    journal = "Phys. Rev. D",
    volume = "93",
    number = "6",
    pages = "064044",
    year = "2016"
}

@article{Minguzzi:2007yq,
    author = "Minguzzi, E.",
    title = "{Limit curve theorems in Lorentzian geometry}",
    eprint = "0712.3942",
    archivePrefix = "arXiv",
    primaryClass = "gr-qc",
    doi = "10.1063/1.2973048",
    journal = "J. Math. Phys.",
    volume = "49",
    pages = "092501",
    year = "2008"
}

@book{Hawking:1973uf,
    author = "Hawking, Stephen W. and Ellis, George F. R.",
    title = "{The Large Scale Structure of Space-Time}",
    doi = "10.1017/9781009253161",
    isbn = "978-1-009-25316-1, 978-1-009-25315-4, 978-0-521-20016-5, 978-0-521-09906-6, 978-0-511-82630-6, 978-0-521-09906-6",
    publisher = "Cambridge University Press",
    series = "Cambridge Monographs on Mathematical Physics",
    month = "2",
    year = "2023"
}

@book{Wald:1984rg,
    author = "Wald, Robert M.",
    title = "{General Relativity}",
    doi = "10.7208/chicago/9780226870373.001.0001",
    publisher = "Chicago Univ. Pr.",
    address = "Chicago, USA",
    year = "1984"
}

@article{Bousso:2022cun,
    author = "Bousso, Raphael and Shahbazi-Moghaddam, Arvin",
    title = "{Singularities from Entropy}",
    eprint = "2201.11132",
    archivePrefix = "arXiv",
    primaryClass = "hep-th",
    doi = "10.1103/PhysRevLett.128.231301",
    journal = "Phys. Rev. Lett.",
    volume = "128",
    number = "23",
    pages = "231301",
    year = "2022"
}

@article{Bousso:2025xyc,
    author = "Bousso, Raphael",
    title = "{Robust Singularity Theorem}",
    eprint = "2501.17910",
    archivePrefix = "arXiv",
    primaryClass = "hep-th",
    doi = "10.1103/6f9b-3jmx",
    journal = "Phys. Rev. Lett.",
    volume = "135",
    number = "1",
    pages = "011501",
    year = "2025"
}

@article{Penrose:1964wq,
    author = "Penrose, Roger",
    title = "{Gravitational collapse and space-time singularities}",
    doi = "10.1103/PhysRevLett.14.57",
    journal = "Phys. Rev. Lett.",
    volume = "14",
    pages = "57--59",
    year = "1965"
}

@article{Bekenstein:1972tm,
    author = "Bekenstein, J. D.",
    title = "{Black holes and the second law}",
    doi = "10.1007/BF02757029",
    journal = "Lett. Nuovo Cim.",
    volume = "4",
    pages = "737--740",
    year = "1972"
}

@article{Maldacena:1997re,
    author = "Maldacena, Juan Martin",
    title = "{The Large $N$ limit of superconformal field theories and supergravity}",
    eprint = "hep-th/9711200",
    archivePrefix = "arXiv",
    reportNumber = "HUTP-97-A097, HUTP-98-A097",
    doi = "10.4310/ATMP.1998.v2.n2.a1",
    journal = "Adv. Theor. Math. Phys.",
    volume = "2",
    pages = "231--252",
    year = "1998"
}

@article{Wall:2010jtc,
    author = "Wall, Aron C.",
    title = "{The Generalized Second Law implies a Quantum Singularity Theorem}",
    eprint = "1010.5513",
    archivePrefix = "arXiv",
    primaryClass = "gr-qc",
    doi = "10.1088/0264-9381/30/19/199501",
    journal = "Class. Quant. Grav.",
    volume = "30",
    pages = "165003",
    year = "2013",
    note = "[Erratum: Class.Quant.Grav. 30, 199501 (2013)]"
}

@article{Penrose:1969pc,
    author = "Penrose, R.",
    title = "{Gravitational collapse: The role of general relativity}",
    doi = "10.1023/A:1016578408204",
    journal = "Riv. Nuovo Cim.",
    volume = "1",
    pages = "252--276",
    year = "1969"
}

@article{Shahbazi-Moghaddam:2022hbw,
    author = "Shahbazi-Moghaddam, Arvin",
    title = "{Restricted quantum focusing}",
    eprint = "2212.03881",
    archivePrefix = "arXiv",
    primaryClass = "hep-th",
    doi = "10.1103/PhysRevD.109.066023",
    journal = "Phys. Rev. D",
    volume = "109",
    number = "6",
    pages = "066023",
    year = "2024"
}

@article{Franken:2025gwr,
    author = "Franken, Victor and Kaya, Sami and Rondeau, Fran{\c{c}}ois and Shahbazi-Moghaddam, Arvin and Tran, Patrick",
    title = "{Tests of restricted Quantum Focusing and a new CFT bound}",
    eprint = "2510.13961",
    archivePrefix = "arXiv",
    primaryClass = "hep-th",
    month = "10",
    year = "2025"
}

@article{Bousso:2019bkg,
    author = "Bousso, Raphael and Shahbazi-Moghaddam, Arvin and Toma{\v{s}}evi{\'c}, Marija",
    title = "{Quantum Information Bound on the Energy}",
    eprint = "1909.02001",
    archivePrefix = "arXiv",
    primaryClass = "hep-th",
    doi = "10.1103/PhysRevD.100.126010",
    journal = "Phys. Rev. D",
    volume = "100",
    number = "12",
    pages = "126010",
    year = "2019"
}

@article{Bousso:2019var,
    author = "Bousso, Raphael and Shahbazi-Moghaddam, Arvin and Tomasevic, Marija",
    title = "{Quantum Penrose Inequality}",
    eprint = "1908.02755",
    archivePrefix = "arXiv",
    primaryClass = "hep-th",
    doi = "10.1103/PhysRevLett.123.241301",
    journal = "Phys. Rev. Lett.",
    volume = "123",
    number = "24",
    pages = "241301",
    year = "2019"
}
\bibliographystyle{utphys1}
\appendix
\section{Sufficient conditions for (A4)}
\label{App:A4}
We now exhibit concrete, easily checked sufficient conditions that imply the Quantum Focussing Blow-up
assumption (A4). The argument follows the standard classical route: derive a differential inequality of the
form \(\partial_\lambda \Theta_{\mathrm{gen}}\le -\kappa\,\Theta_{\mathrm{gen}}^2\) with \(\kappa>0\), and conclude
finite-affine blow-up by integrating the inequality. We use the decomposition
\begin{equation}\label{eq:Theta-decomp}
\Theta_{\mathrm{gen}}=\theta + q\,S'_{\rm out},\qquad q:=\frac{4G\hbar}{A},
\end{equation}
(cf. Eq.~(5.6) of \cite{Bousso:2015mna}).

\begin{lemma}[Sufficient conditions for (A4)]
Let \(\Sigma\) be a compact surface and let \(k^a\) be a future-directed null normal. Fix a neighbourhood
in which \(S_{\rm gen}\) is \(C^2\) along generators and the semiclassical description holds. Suppose there
exist nonnegative constants \(\alpha,\beta,\gamma\) such that along the generator under consideration
the following bounds hold at all points of the neighbourhood:
\begin{align}
|S'_{\rm out}| &\le \alpha\,\frac{A}{4G\hbar}\,|\theta|, \label{eq:bound1}\\
|S''_{\rm out}| &\le \beta\,\theta^2, \label{eq:bound2}\\
R_{ab}k^a k^b &\ge -\gamma\,\theta^2, \label{eq:bound3}
\end{align}
and that the constants satisfy the strict inequality
\begin{equation}\label{eq:alpha-beta-gamma}
\alpha+\beta+\gamma < \frac{1}{n}, \qquad n:=D-2.
\end{equation}
Then there exists \(\kappa>0\) (explicit below) such that along the generator
\[
\partial_\lambda \Theta_{\mathrm{gen}} \le -\kappa\,\Theta_{\mathrm{gen}}^2,
\]
in particular, if \(\Theta_{\mathrm{gen}}(\lambda_0)<0\) then \(\Theta_{\mathrm{gen}}\) blows up to \(-\infty\)
within affine parameter no larger than \(\lambda_0 + 1/(\kappa|\Theta_{\mathrm{gen}}(\lambda_0)|)\). Hence (A4)
holds with \(C\) and \(\Lambda\) determined by \(\kappa\) and \(\Theta_{\mathrm{gen}}(\lambda_0)\).
\end{lemma}
\begin{proof}
Differentiate \eqref{eq:Theta-decomp} along the generator. Using \(q'=-q\theta\) (since \(A'/A=\theta\)) one obtains
\begin{equation}\label{eq:Theta-prime}
\Theta'_{\mathrm{gen}} = \theta' + q\big(S''_{\rm out} - \theta S'_{\rm out}\big).
\end{equation}
The classical Raychaudhuri equation gives
\[
\theta' = -\frac{1}{n}\theta^2 - \sigma^2 - R_{ab}k^a k^b \le -\frac{1}{n}\theta^2 - R_{ab}k^a k^b,
\]
where we discarded the nonpositive shear term \(-\sigma^2\). Combining with \eqref{eq:Theta-prime} and using
the bound \eqref{eq:bound3} yields
\begin{equation}\label{eq:intermediate}
\Theta'_{\mathrm{gen}} \le \Big(-\frac{1}{n} + \gamma\Big)\theta^2 + q\big(S''_{\rm out} - \theta S'_{\rm out}\big).
\end{equation}
Now apply the bounds \eqref{eq:bound1}--\eqref{eq:bound2}. Using \(q( A/(4G\hbar))=1\) we have
\[
|q\,\theta S'_{\rm out}| \le q\,|\theta|\,|S'_{\rm out}| \le \alpha\,\theta^2,
\qquad
|q\,S''_{\rm out}| \le q\,\beta\,\theta^2 = \beta\,\theta^2,
\]
hence
\[
q\big(S''_{\rm out} - \theta S'_{\rm out}\big) \le (\beta+\alpha)\,\theta^2.
\]
Substituting into \eqref{eq:intermediate} gives
\[
\Theta'_{\mathrm{gen}} \le \Big(-\frac{1}{n} + \gamma + \beta + \alpha\Big)\theta^2 .
\]
Use \(\Theta_{\mathrm{gen}}=\theta + qS'_{\rm out}\) and the bound \(|qS'_{\rm out}|\le \alpha|\theta|\) to relate \(|\theta|\)
and \(|\Theta_{\mathrm{gen}}|\). Indeed
\[
|\Theta_{\mathrm{gen}}| \ge |\theta| - |qS'_{\rm out}| \ge (1-\alpha)|\theta|,
\]
so \(|\theta|\le |\Theta_{\mathrm{gen}}|/(1-\alpha)\) provided \(\alpha<1\). (If \(\alpha\ge1\) the hypothesis \eqref{eq:alpha-beta-gamma}
already fails because the right-hand side of \eqref{eq:alpha-beta-gamma} is \(1/n\le 1\).) Therefore
\[
\theta^2 \le \frac{1}{(1-\alpha)^2}\,\Theta_{\mathrm{gen}}^2.
\]
Combining yields
\[
\Theta'_{\mathrm{gen}} \le -\underbrace{\frac{1/n - (\alpha+\beta+\gamma)}{(1-\alpha)^2}}_{:=\kappa}\;
\Theta_{\mathrm{gen}}^2.
\]
By the strict inequality \eqref{eq:alpha-beta-gamma} we have \(\kappa>0\). The standard differential inequality
\(\partial_\lambda \Theta \le -\kappa\Theta^2\) with \(\Theta(\lambda_0)<0\) implies (via separation of variables)
\[
\frac{1}{|\Theta(\lambda)|} \le \frac{1}{|\Theta(\lambda_0)|} - \kappa(\lambda-\lambda_0),
\]
so \(|\Theta(\lambda)|\to\infty\) at latest at
\(\lambda = \lambda_0 + \frac{1}{\kappa|\Theta(\lambda_0)|}\). This produces a focal (conjugate) point in finite affine
parameter, establishing (A4) with explicit constants determined by \(\alpha,\beta,\gamma\) and \(\Theta(\lambda_0)\).
\end{proof}
The three bounds \eqref{eq:bound1}--\eqref{eq:bound3} are physically transparent:
\(\alpha\) controls the relative size of the entropy flux to the geometric expansion, \(\beta\) bounds the entropy curvature,
and \(\gamma\) bounds Ricci focusing relative to \(\theta^2\). Near strong-curvature semiclassical regimes \(|\theta|\) typically
grows large and the area term in \(S_{\mathrm{gen}}\) dominates the entropy derivatives, so one expects \(\alpha,\beta,\gamma\)
to be small and \eqref{eq:alpha-beta-gamma} to hold. Thus the lemma provides a clean, checkable sufficient condition for (A4).
For background on the decomposition \(\Theta_{\mathrm{gen}}=\theta + (4G\hbar/A)S'_{\rm out}\) and the behaviour of the entropy
terms see Bousso \emph{et al.}. Also see the quantum-entropic obstruction proved earlier in this work.

\end{document}